\documentclass[11pt,a4paper]{article}

\usepackage[a4paper,left=2.9cm,right=2.9cm,top=3cm,bottom=4cm]{geometry}

\newcommand{\footremember}[2]{%
    \footnote{#2}
    \newcounter{#1}
    \setcounter{#1}{\value{footnote}}%
}
\newcommand{\footrecall}[1]{%
    \footnotemark[\value{#1}]%
} 

\usepackage{amsmath}
\usepackage{amssymb}
\usepackage{amsthm}

\newcommand{\eqdef}{\mathrel{\mathop{\raisebox{1pt}{\scriptsize$:$}}}=}
\usepackage{tikz}
\usetikzlibrary{arrows,automata,shapes,shapes.multipart,positioning,calc}
\usepackage{todonotes}
\newcommand{\class}[1]{\langle #1 \rangle}
\newcommand{\curly}[1]{\left\lbrace#1\right\rbrace}
\newcommand{\set}[1]{\curly{#1}}
\newcommand{\where}{\,\middle|\,}
\newcommand{\subcl}[1]{\mathcal{SUB}\left(#1\right)}
\newcommand{\focl}[1]{\mathcal{FO}\left(#1\right)}

\DeclareMathOperator{\rep}{rep}
\makeatletter

\newtheorem{definition}{Definition}
\newtheorem{theorem}{Theorem}
\newtheorem{lemma}{Lemma}

\newtheorem{example}{Example}

\title{Representing Pattern Matching Algorithms by Polynomial-Size Automata}
\author{Tobias Marschall\footremember{cbi}{Center for Bioinformatics, Saarland University, Saarbr\"ucken, Germany}\footnote{Max Planck Institute for Informatics, Saarbr\"ucken, Germany} \and Noemi E. Passing\footrecall{cbi} \footnote{Saarbr\"ucken Graduate School of Computer Science, Saarland University, Saarbr\"ucken, Germany}}
\date{}

\begin{document}
\maketitle
	
\begin{abstract}
Pattern matching algorithms to find exact occurrences of a pattern $S\in\Sigma^m$ in a text $T\in\Sigma^n$ have been analyzed extensively with respect to asymptotic best, worst, and average case runtime.
For more detailed analyses, the number of text character accesses $X^{\mathcal{A},S}_n$ performed by an algorithm $\mathcal{A}$ when searching a random text of length $n$ for a fixed pattern $S$ has been considered.
Constructing a state space and corresponding transition rules (e.g. in a Markov chain) that reflect the behavior of a pattern matching algorithm is a key step in existing analyses of $X^{\mathcal{A},S}_n$ in both the asymptotic ($n\to\infty$) and the non-asymptotic regime.
The size of this state space is hence a crucial parameter for such analyses.
In this paper, we introduce a general methodology to construct corresponding state spaces and demonstrate that it applies to a wide range of algorithms, including Boyer-Moore (BM), Boyer-Moore-Horspool (BMH), Backward Oracle Matching (BOM), and Backward (Non-Deterministic) DAWG Matching (B(N)DM).
In all cases except BOM, our method leads to state spaces of size $O(m^3)$ for pattern length $m$, a result that has previously only been obtained for BMH.
In all other cases, only state spaces with size exponential in $m$ had been reported.
Our results immediately imply an algorithm to compute the distribution of $X^{\mathcal{A},S}_n$ for  fixed~$S$, fixed~$n$, and $\mathcal{A}\in\{\text{BM},\text{BMH},\text{B(N)DM}\}$ in polynomial time for a very general class of random text models.
\end{abstract}

\section{Introduction}
Algorithms to swiftly find exact occurrences of a pattern $S\in\Sigma^m$ in a text $T\in\Sigma^n$ over a finite alphabet $\Sigma$ are a classical topic in computer science and have been studied for decades.
Prominent examples include the Boyer-Moore algorithm (BM) \cite{BoyerM1977}, Boyer-Moore-Horspool (BMH) algorithm \cite{Horspool1980}, and the Knuth-Morris-Pratt (KMP) algorithm \cite{KnuthMP1977}.
The KMP algorithm processes the input text from left to right, accesses every text character exactly once, and uses amortized constant time per character for an optimal worst-case runtime of $O(n+m)$.
The runtime of KMP is independent of input text and pattern, so its best case behavior matches its worst case behavior.
The BMH algorithm, in contrast, achieves a best-case runtime of $O(n/m)$, but exhibits a worst-case runtime of $O(nm)$.

BM and BMH are just two examples of \emph{window-based pattern matching algorithms} that consider a length-$m$ window in the text, compare it to the pattern, and then shift the window by a number of positions.
Under favorable circumstances, the window can be shifted by more than one position, leading to good best-case behavior.
The BMH algorithm, for example, compares window content and pattern from right to left and determines the shift based on the last character in the window.
In case this character is absent from the pattern, the window can safely be shifted by $m$ positions without missing a pattern occurrence.
Backward Oracle Matching (BOM) \cite{AllauzenCR2001}, Backward DAWG Matching (BDM) \cite{CrochemoreCGJLPR1994}, Backward Nondeterministic DAWG Matching (BNDM) \cite{Navarro1998}, or Sunday's algorithm \cite{Sunday1990} constitute further examples for window-based pattern matching algorithms.
While following the same general idea, they differ in how the shift is determined, which data structures are needed to determine it, and consequently how long it takes to construct these data structures in a pre-processing step.

In order to study the runtime of such algorithms beyond asymptotic best, average, and worst case behavior, we consider the number of accesses to characters in the input text.
Let $X^{\mathcal{A},S}_n$ denote the random variable giving the number of character accesses done by algorithm $\mathcal{A}$ while searching a random text of length $n$ for pattern $S\in\Sigma^m$.
For the BMH algorithm, i.e.\ when $\mathcal{A}=\text{BMH}$, a number of results on the properties of $X^{\mathcal{A},S}_n$ exist.
In early work \cite{Baeza-Yates1990,Baeza-Yates1992}, it has been shown that $\frac{E[X^{\text{BMH},S}_n]}{n}$  approaches a constant as $n$ goes to infinity.
The same articles provide methods to compute this constant for a given pattern and to obtain its average over all patterns of a given length.
These results were obtained for the case that all characters in the random text are drawn independently with uniform probability.
Later, they have been extended \cite{Mahmoud1997} to show that the distribution of $X^{\text{BMH},S}_n$, written $\mathcal{L}\left(X^{\text{BMH},S}_n\right)$, converges to a normal distribution for $n\to\infty$ (after $X^{\text{BMH},S}_n$ has been properly centered and normalized).
Further generalizations have yielded the same result for Markovian text models \cite{Smythe2001}, i.e.\ when the probability of a character depends on the previous character, and have shown how to compute the asymptotic variance for BMH as well as for BM \cite{Tsai2006}.
Otherwise, all these results are restricted to BMH and do not easily extend to other algorithms.
Conceptually, all these analyses are based on the idea of constructing a suitable set of \emph{states} capturing the algorithms' status and by then building a Markov chain on this state space to derive the desired quantities.

Constructing such a state space can be simplified by employing the framework of \emph{deterministic arithmetic automata (DAAs)} and \emph{probabilistic arithmetic automata (PAAs)} \cite{Marschall2012,MarschallR2011}.
This approach decouples the description of the algorithms' behavior (encoded as a DAA) from the specification of a text model, which together give rise to a PAA.
Using this technique, PAAs for the algorithms BM, BMH, BDM, BNDM, and BOM have been constructed for arbitrary \emph{finite-memory text models} and subsequently been used to compute the full distribution $\mathcal{L}\left(X^{\mathcal{A},S}_n\right)$ for fixed $n$, see \cite{MarschallR2011}.

Asymptotic analyses (such as \cite{Tsai2006}) as well as non-asymptotic analyses (such as \cite{MarschallR2011}) both hinge on finding a suitable state space (and the same state space can be used for both purposes).
The state space given in \cite{Tsai2006} has a size of $O(m^3)$, but is tied to the BMH algorithm and to first-order Markovian text models (although the author states that it can be generalized to higher-order Markovian text models and to the BM algorithm).
The approach in \cite{MarschallR2011} allows to construct a state space for a wide range of window-based pattern matching algorithms and applies to arbitrary text models, but leads to $O(\Sigma^m\cdot m)$ states.
Although \cite{MarschallR2011} also provide an algorithm to minimize the state space, similar to DFA minimisation, the exponential state space needs to be constructed first.

\paragraph*{Contributions}
Here, we report on a direct and general procedure of constructing small DAAs.
Applied to BM, BMH and B(N)DM, it yields DAAs with $O(m^3)$ states.
This immediately implies an algorithm to compute $\mathcal{L}\left(X^{\mathcal{A},S}_n\right)$ in polynomial time for arbitrary finite-memory text models.
Our construction is also applicable to BOM, but without getting a respective guarantee on the state space size (which depends on the details of the factor oracle in this case).
In all cases, we obtain a PAA (which can be interpreted as a Markov chain), that can be used to obtain asymptotic results like existence and value of  $\lim_{n\to\infty}\frac{E[X^{\text{BMH},S}_n]}{n}$.
We emphasize that the methodology we introduce is general in nature and can likely be applied to analyze further pattern matching algorithms with little extra work.

\paragraph*{Article Organization}
We start by briefly revisiting the BM, BMH, B(N)DM, and BOM algorithms in Section~\ref{sec:algorithms} and proceed by explaining the DAA/PAA framework developed by \cite{Marschall2012,MarschallR2011} in Section~\ref{sec:arithmetic_automata}.
Our main results are presented in Section~\ref{ch:scheme}, where we devise a general method to construct small DAAs and prove that it applies to BM, BMH, B(N)DM, and BOM.
Section~\ref{ch:PAA} explores some of the consequences of these results and Section~\ref{sec:conclusion} provides a concluding discussion.

\section{Pattern Matching Algorithms}\label{sec:algorithms}

In this section, we summarize the BM, BMH, B(N)DM, and BOM algorithms. They all maintain a search window $w$ of length of the pattern $|S|=m$ that initially starts at position $0$ in text $T$ such that its rightmost character is at position $m-1$. The window position increases in the course of the algorithm. The two properties of an algorithm $\mathcal{A}$ that influence our analysis are the cost $f^S_\mathcal{A}(w)$ of the window, i.e.\ the number of character accesses required to analyze the window, and the shift $g^S_\mathcal{A}(w)$ of the window, i.e.\ the number of characters the window is shifted after it has been examined.

\subsection{Boyer-Moore algorithm}

The Boyer-Moore (BM) algorithm \cite{BoyerM1977} compares a window $w=w_0 \dots w_{m-1}$ and a pattern $S=s_0 \dots s_{m-1}$ from right to left until it reaches the first mismatch or verifies a complete match of pattern and window. Thus, the cost of $w$ is
\begin{equation*}
	f^S_{\text{BM}}(w) \eqdef \begin{cases}m & \text{if } w=S,\\
	\min_{1 \leq i \leq m}\set{i \where w_{m-i} \neq s_{m-i}} & \text{otherwise}.\end{cases}
\end{equation*}
If the mismatch occurs at position $m-i$, the BM algorithm uses two different rules to determine the shift:
The \emph{bad-character rule} aligns $w_{m-i}$ with the rightmost occurrence of the same character in $s_0 \dots s_{m-i-1}$ if existing. Otherwise, we shift by $m-i+1$ characters:
\begin{equation*}
	bc^S(w,i) \eqdef \begin{cases}m -i +1 & \text{if } w_{m-i} \notin \set{s_0,\dots,s_{m-i-1}},\\
	\min_{1 \leq k \leq m-i}\set{k \where w_{m-i}=s_{m-i-k}} & \text{otherwise}.\end{cases}
\end{equation*}
The \emph{good-suffix rule} aligns $w_{m-i+1} \dots w_{m-1}$ with the rightmost occurrence $s_k \dots s_{k+i-2}$ of the same string in $s_0 \dots s_{m-2}$ with $s_{k-1} \neq w_{m-i}$.
The set of positions where such strings start in $S$ is given by
\[\mathcal{S}^S(w,i)\eqdef\big\{k\in\{0,\dots,m-i\}\big|s_k\dots s_{k+i-2}=w_{m-i+1}\dots w_{m-1} \text{ and } s_{k-1} \neq w_{m-i} \big\}.\]
If no such position exists, i.e.\ $\mathcal{S}^S$ is empty, then it shifts by the least amount such that a prefix of $S$ becomes aligned with a suffix of $w_{m-i+1} \dots w_{m-1}$.
In other words, it aligns $s_0\dots s_k$ with $w_{m-k-1} \dots w_{m-1}$ for the largest possible $k$.
Or formally
\[\mathcal{K}^S(w,i)\eqdef\big\{k\in\{0,\dots,i-2\}\big|s_0\dots s_k = w_{m-k-1} \dots w_{m-1}\big\}\]
and
\[gs^S(w,i) \eqdef
\begin{cases}
m-i-\max\left(\mathcal{S}^S(w,i)\right)+1 & \text{if } \mathcal{S}^S(w,i)\neq\emptyset,\\
m-\max\left(\mathcal{K}^S(w,i)\right) -1  & \text{if } \mathcal{S}^S(w,i)=\emptyset\text{ and }\mathcal{K}^S(w,i)\neq\emptyset,\\
m                                       & \text{otherwise}.
\end{cases}\]
The shift of the BM algorithm is the maximum of the shifts determined by bad-character and good-suffix rules:
\begin{equation*}
g^S_{\text{BM}}(w) \eqdef \max\set{gs^S(w,i), bc^S(w,i)}.
\end{equation*}

\subsection{Boyer-Moore-Horspool algorithm}

The Boyer-Moore-Horspool (BMH) algorithm \cite{Horspool1980} is a simplification of the Boyer-Moore algorithm that uses only a bad-character rule using the rightmost character of the window instead of the mismatch character. Thus, the cost of a window is the same as for the BM algorithm and the shift is $g^S_{\text{BMH}}(w) \eqdef bc^S(w,1)$.

\subsection{Backward (Non-Deterministic) DAWG Matching algorithm}

The Backward DAWG Matching (BDM) algorithm \cite{CrochemoreCGJLPR1994,CrochemoreR1994} is based on a suffix automaton that accepts all suffixes of the reversed pattern $S^{\text{rev}}=s_{m-1} \dots s_0$ and enters a FAIL-state if a string has been read that is not a substring of $S^{\text{rev}}$.
Similar to the BMH algorithm, the windows are processed from right to left. 
The number of character accesses and hence the cost of a window $w$ equals the number of transitions the suffix automaton of $S^{\text{rev}}$ does on $w$. 
Let $j^S_w$ be the number of transitions the suffix automaton of $S^\text{rev}$ does before entering the FAIL-state (excluding the transition to the FAIL-state) when reading $w$. Then
\begin{equation*}
	f^S_{\text{BDM}}(w) \eqdef \begin{cases}m & \text{if } w=S\\
	j^S_w+1 & \text{otherwise}\end{cases}
\end{equation*}
We define a set $\mathcal{I}^{S}(w) \subseteq \set{0, \dots, m-1}$ with $i \in \mathcal{I}^S(w)$ if and only if the suffix automaton of $S^{\text{rev}}$ is in an accepting state after reading $i$ characters of $w^{rev}$. Hence, $\mathcal{I}^{S}(w)$ is the set that contains all positions $i$ such that $w_{m-i} \dots w_{m-1}$ is a prefix of $S$. The shift function is
\begin{equation*}
	g^S_{\text{BDM}}(w) \eqdef \min\set{m-i \where i \in \mathcal{I}^S(w)},
\end{equation*}
thus aligning a matching suffix of the current window with the prefix of the next window.

The Backward Non-Deterministic DAWG Matching (BNDM) algorithm \cite{NavarroR1998} is a version of the BDM algorithm that uses a nondeterministic automaton instead of a deterministic suffix automaton. The construction of the non-deterministic automaton is much easier, but the processing of a text character takes $O(m)$ time instead of $O(1)$ time. However, the nondeterministic automaton can be implemented with bit-parallel operations efficiently, which is especially beneficial when $m$ is smaller than the machine's word size. We refer to \cite{NavarroR1998} for the explicit construction.
Since the BDM and BNDM only differ in the implementation of the automation but not in cost and shift functions, we do not distinguish between them in the following.
	
\subsection{Backward Oracle Matching algorithm}

The Backward Oracle Matching (BOM) algorithm \cite{AllauzenCR1999, AllauzenCR2001} is a variant of the BDM algorithm that uses a simpler deterministic automaton, the so-called factor oracle, for the reversed pattern $S^{\text{rev}}$. 
The simplicity (and fast construction) comes with the downside that the BOM algorithm recognizes (slightly) more strings than the substrings of the pattern, which can lead to unnecessary comparisons and sub-optimal shifts.
We refer to \cite{AllauzenCR1999} for its construction and a detailed discussion.
The algorithm processes the windows from right to left and defines the cost of a window $w$ as the number of transitions the factor oracle does on $w$ including the transition to the FAIL-state. Let $k^S_w$ be the number of transitions the factor oracle of $S^\text{rev}$ does before entering the FAIL-state (excluding the transition to the FAIL-state) when reading $w$. We define the cost function and shift function as follows:
\begin{equation*}
	f^S_{\text{BOM}}(w) \eqdef \begin{cases}m & \text{if } w=S,\\
	k^S_w+1 & \text{otherwise},\end{cases}
\end{equation*}
\begin{equation*}
	g^S_{\text{BOM}}(w) \eqdef \begin{cases}1 & \text{if } w = S,\\
	m-k^S_w & \text{otherwise}.\end{cases}
\end{equation*}
One can observe that the BOM algorithm uses smaller shifts and thus has to consider more windows than the BDM algorithm. However, a factor oracle has only $m+1$ states \cite{AllauzenCR1999} and can be constructed faster than a suffix automaton. Thus, the BOM algorithm uses less space and needs less preprocessing time.

\section{Arithmetic Automata Approach}\label{sec:arithmetic_automata}
The framework of deterministic arithmetic automata (DAA) and probabilistic arithmetic automata (PAA)  allows to conveniently model computations performed sequentially on random sequences.
In \cite{Marschall2012}, this framework is reviewed in depth and many applications are discussed.
It was previously applied to compute the character access count distribution $\mathcal{L}\left(X^{\mathcal{A},S}_n\right)$ for BMH, B(N)DM, and BOM \cite{MarschallR2011}.

\subsection{Deterministic Arithmetic Automata}
%
%

Deterministic arithmetic automata (DAAs) are a variant of deterministic finite automata (DFAs) that additionally perform an arithmetic operation at each state transition.
For example, they can count how often the automaton visits a specific state while reading a text. 

\begin{definition}[Deterministic Arithmetic Automaton, \cite{Marschall2012}]
	A \emph{deterministic arithmetic automaton (DAA)} is a tuple
	$\mathcal{D} = \left( \mathcal{Q}, q_{0}, \Sigma, \delta, \mathcal{V}, v_{0}, \mathcal{E}, \left(\eta_{q}\right)_{q \in \mathcal{Q}}, \left(\theta_{q}\right)_{q \in \mathcal{Q}}\right)$,
	where $\mathcal{Q}$ denotes a finite set of states, $q_{0} \in \mathcal{Q}$ is the start  state, $\Sigma$ is a finite alphabet, $\delta: \mathcal{Q} \times \Sigma \rightarrow \mathcal{Q}$ is a transition function, $\mathcal{V}$ is a set of values, $v_{0} \in \mathcal{V}$ is called the start value, $\mathcal{E}$ is a finite set of emissions, $\eta_{q} \in \mathcal{E}$ is the emission associated with state $q$ and $\theta_{q}: \mathcal{V} \times \mathcal{E} \rightarrow \mathcal{V}$ is a binary operation associated with state $q$. 
\end{definition}

Intuitively, a DAA starts with the pair $\left(q_{0},v_{0}\right)$ and reads a sequence of characters from the alphabet. When the automaton is in state $q$ with value $v$ and reads a character $\sigma \in \Sigma$, it does a transition to state $q' \eqdef \delta\left(q,\sigma\right)$ and updates its value to $v' \eqdef \theta_{q'}\left(v,\eta_{q'}\right)$ using the binary operation $\theta_{q'}$ and the emission $\eta_{q'}$ of the successor state $q'$.
To observe the transitions in the states and the updates of the values simultaneously, we define the \emph{joint transition function} of deterministic arithmetic automata.

\begin{definition}[Joint transition function, \cite{Marschall2012}]\label{def:joint transition function} The associated \emph{joint transition function} $\tilde{\delta}: \left( \mathcal{Q} \times \mathcal{V} \right) \times \Sigma \rightarrow \left( \mathcal{Q} \times \mathcal{V} \right)$ of a DAA $\mathcal{D} = \left( \mathcal{Q}, q_{0}, \Sigma, \delta, \mathcal{V}, v_{0}, \mathcal{E}, \left(\eta_{q}\right)_{q \in \mathcal{Q}}, \left(\theta_{q}\right)_{q \in \mathcal{Q}}\right)$ is defined by:
	\begin{equation*}
		\tilde{\delta}\left(\left(q,v\right),\sigma\right) \eqdef \left( \delta\left(q,\sigma\right), \theta_{\delta\left(q,\sigma\right)}\left(v, \eta_{\delta\left(q,\sigma\right)}\right)\right)
	\end{equation*}
	We extend this definition inductively from $\Sigma$ to $\Sigma^*$ by defining $\hat{\delta}:  \left( \mathcal{Q} \times \mathcal{V} \right) \times \Sigma^* \rightarrow \left( \mathcal{Q} \times \mathcal{V} \right):$
	\begin{align*}
		\hat{\delta}\left(\left(q,v\right),\varepsilon\right) &\eqdef \left(q,v\right) \\
		\hat{\delta}\left(\left(q,v\right),x\sigma\right) &\eqdef \tilde{\delta}\left(\hat{\delta}\left(\left(q,v\right),x\right),\sigma\right)
	\end{align*}
\end{definition}
Note that, for a finite value sets~$\mathcal{V}$, a DAA can hence be be seen as a DFA over the state space $\mathcal{Q}\times\mathcal{V}$.
So DAAs do not provide more expressive power; their benefit lies in allowing for cleaner and more intuitive models.

\begin{definition}[Value computed by a DAA, \cite{Marschall2012}]
	Let $\mathcal{D}$ be a \emph{DAA} with transition function $\delta$, let $\hat{\delta}$ be the joint transition function of $\mathcal{D}$ and let $T \in \Sigma^m$ be a text. We say that $\mathcal{D}$ computes the value $v$ if we have $\hat{\delta}\left(\left(q_{0},v_{0}\right),T\right) = \left(q,v\right)$ for some $q \in \mathcal{Q}$. We define $value_{\mathcal{D}}\left(T\right) \eqdef v$.
\end{definition}

\paragraph*{DAAs encoding Pattern Matching Algorithms}

We construct a DAA that upon reading a text $T \in \Sigma^*$ calculates the total cost of $T$, i.e. the number of character accesses an algorithm $\mathcal{A}$ with cost function $f^{S}_{\mathcal{A}}$ and shift function $g^{S}_{\mathcal{A}}$ performs when searching $T$ for a pattern $S$.

\begin{definition}[DAA encoding a pattern matching algorithm, \cite{MarschallR2011}]\label{def:DAA encoding pattern matching algorithm}
	Let $S \in \Sigma^m$ be a pattern and let $\mathcal{A}$ be an algorithm with functions $f^{S}_{\mathcal{A}}$ and $g^{S}_{\mathcal{A}}$. We define the DAA $\mathcal{D}$ that encodes $\mathcal{A}$ as
	\begin{itemize}
		\item $\mathcal{Q} \eqdef \Sigma^m \times \set{0, \dots ,m}$
		\item $q_{0} \eqdef (S,m)$
		\item $\delta\left((w,k),\sigma\right) \eqdef \begin{cases} \left(w_{1} \dots w_{m-1}\sigma,k-1\right) & \text{if }k>0\\\left(w_{1} \dots w_{m-1}\sigma, g^S_\mathcal{A}(w)-1\right) & \text{if }k=0\end{cases}$
		\item $\mathcal{V} \eqdef \mathbb{N}$
		\item $v_{0} \eqdef 0$
		\item $\mathcal{E} \eqdef \set{1, \dots ,m}$
		\item $\eta_{\left(w,k\right)} \eqdef \begin{cases}0 & \text{if } k>0\\f^{S}_{\mathcal{A}}(w) & \text{if } k=0\end{cases} ~ ~ ~ ~ ~ \forall ~ (w,k) \in \mathcal{Q}$
		\item $\theta_{\left(w,k\right)}\left(v,e\right) \eqdef v+e ~ ~ ~ ~ ~ \forall ~ (w,k) \in \mathcal{Q}$
	\end{itemize}
\end{definition}

Note that we represent states as pairs of strings and natural numbers. 
The transition function is constructed such that, for a state $(w,k)$, the string $w$ corresponds to the last $m$ characters read by the automaton. The integer $k$ gives the number of characters that the automaton has to read until it reaches the end of the current window.
That is, $k=0$ if and only if the pattern matching algorithm would process the window that ends at the last-read text character.
Figure \ref{fig:DAA behavior} shows an example of how a DAA for the BMH algorithm moves from state to state.
	\begin{figure}
		\begin{center}
			\scalebox{0.9}{
				\begin{tikzpicture}[thick]
				  \tikzstyle{every state}=[fill=none,draw=black,text=black]
					\node[draw,minimum size=0.7cm]	(A)				at (0,0)	{\strut a};
					\node[draw,minimum size=0.7cm]	(B)				at (0.7,0)	{\strut b};
					\node[draw,minimum size=0.7cm]	(C)				at (1.4,0)	{\strut b};
					\node[draw,minimum size=0.7cm]	(D)				at (2.1,0)	{\strut a};
					\node[draw,minimum size=0.7cm]	(E)				at (2.8,0)	{\strut a};
					
					\node[draw=none,rotate=75] (I)		at (-0.7,1.1)	{$(aa,2)$};
					\node[draw=none,rotate=75]	(A1)	at (0,1.1)	{$(aa,1)$};
					\node[draw=none,rotate=75]	(B1)	at (0.7,1.1)	{$(ab,0)$};
					\node[draw=none,rotate=75]	(C1)	at (1.4,1.1)	{$(bb,1)$};
					\node[draw=none,rotate=75]	(D1)	at (2.1,1.1)	{$(ba,0)$};
					\node[draw=none,rotate=75]	(E1)	at (2.8,1.1)	{$(aa,0)$};
					
					\draw (-0.35,-0.7) -- (1.05,-0.7);
					\draw (-0.35,-0.7) -- (-0.35,-0.6);
					\draw (1.05,-0.7) -- (1.05,-0.6);
					
					\draw (1.05,-0.9) -- (2.45,-0.9);
					\draw (1.05,-0.9) -- (1.05,-0.8);
					\draw (2.45,-0.9) -- (2.45,-0.8);
					
					\draw (1.75,-1.1) -- (3.15,-1.1);
					\draw (1.75,-1.1) -- (1.75,-1);
					\draw (3.15,-1.1) -- (3.15,-1);
					
					\node[draw=none,align=left,text width=2cm]	at (-3,1.1)	{DAA state:};
					\node[draw=none,align=left,text width=2cm]	at (-3,0)		{Text:};
					\node[draw=none,align=left,text width=2cm]	at (-3,-0.8)	{Windows:};
					
					\node[draw=none,align=left,text width=5cm] at (8,1) {$g^{aa}_{\text{BMH}}(aa)=g^{aa}_{\text{BMH}}(ba)=1$};
					\node[draw=none,align=left,text width=5cm] at (8,0.5) {$g^{aa}_{\text{BMH}}(bb)=g^{aa}_{\text{BMH}}(ab)=2$};
					\node[draw=none,align=left,text width=5cm] at (8,0) {$f^{aa}_{\text{BMH}}(aa)=f^{aa}_{\text{BMH}}(ba)=2$};
					\node[draw=none,align=left,text width=5cm] at (8,-0.5) {$f^{aa}_{\text{BMH}}(bb)=f^{aa}_{\text{BMH}}(ab)=1$};
							
				\end{tikzpicture}}
		\end{center}
		\caption{Illustration of the behavior of the BMH algorithm when searching the text $T=abbaa$ for the pattern $S=aa$. On top, one sees the state the DAA takes after reading the character below. At the bottom, the windows considered by the BMH algorithm are indicated.}\label{fig:DAA behavior}
	\end{figure}
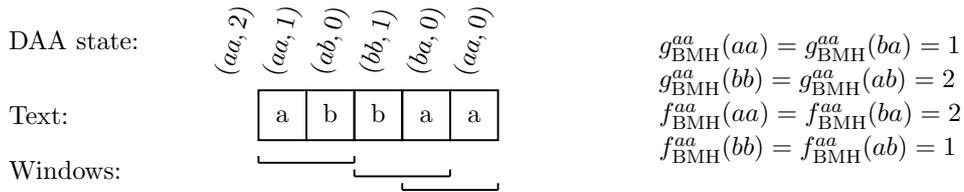
In \cite{MarschallR2011} it is proven that a DAA constructed according to Definition \ref{def:DAA encoding pattern matching algorithm} for an algorithm $\mathcal{A}$ computes the number of character accesses of $\mathcal{A}$ on a text $T$ for a pattern $S$:

\begin{lemma}[DAA correctness, \cite{MarschallR2011}]\label{lem:DAA calculates number of character comparisons}
	Let $\mathcal{A}$ be a pattern matching algorithm and let $\mathcal{D}$ be the DAA encoding $\mathcal{A}$ constructed according to Definition \ref{def:DAA encoding pattern matching algorithm} for pattern $S \in \Sigma^m$. Then $ value_{\mathcal{D}}\left(T\right) = f^S_{\mathcal{A}}\left(T\right)$ for all $T \in \Sigma^n$, where $f^S_\mathcal{A}(T)$ denotes the total cost incurred on $T$.
\end{lemma}

\subsection{Finite-Memory Text Models}
We seek to study the distribution of $value_{\mathcal{D}}\left(T\right)$ when $T$ is a random text.
A random text is a stochastic process $\left(T_i\right)_{i \in \mathbb{N}_0}$, where each $T_i$ takes values in a finite alphabet $\Sigma$. 
A prefix  $T_0 \dots T_{n-1}$ is called a random text of length $n$. 
A text model specifies the probability $\Pr\left[T_0 \dots T_{|x|-1} = x\right]$ for all strings $x \in \Sigma^*$.
One way of doing this is the following.

\begin{definition}[Finite-memory text model]
	A \emph{finite-memory text model} is a tuple $\left(\mathcal{C},c_0,\Sigma,\varphi\right)$, where $\mathcal{C}$ is a finite state space (called context space), $c_0 \in \mathcal{C}$ is a start context, $\Sigma$ is an alphabet and $\varphi: \mathcal{C} \times \Sigma \times \mathcal{C} \rightarrow \left[0,1\right]$ is a transition function with $\sum_{\sigma \in \Sigma, c' \in \mathcal{C}}{\varphi\left(c,\sigma,c'\right)}=1$ for all $c \in \mathcal{C}$. 
\end{definition}

Intuitively, a finite-memory text model $\mathcal{M} = \left(\mathcal{C}, c_0,\Sigma,\varphi\right)$ generates a random text by moving from context to context and emitting a character at each transition. Here, $\varphi\left(c,\sigma,c'\right)$ denotes the probability that $\mathcal{M}$ does a transition from $c$ to $c'$ and thereby generates the character $\sigma$.
In this way the stochastic process $\left(T_i\right)_{i \in \mathbb{N}_0}$ is fully defined (by Kolmogorov's extension theorem).
Markovian models of arbitrary order as well as character-emitting hidden Markov models (HMMs) can be expressed as a finite-memory text model.

\subsection{Probabilistic Arithmetic Automata}
Next, we want to combine a DAA with a finite-memory text model in order to study the DAA's behavior on random texts.
To this end, we use the probabilistic counterpart of DAAs, namely probabilistic arithmetic automata (PAAs).
PAAs can be viewed as (discrete-time) hidden Markov models that additionally perform arithmetic operations on the emission sequence.


\begin{definition}[Probabilistic Arithmetic Automaton, \cite{Marschall2012}]
	A \emph{probabilistic arithmetic automaton (PAA)} is a tuple
	$\mathcal{P} = \left( \mathcal{Q}, q_{0}, \delta_T, \mathcal{V}, v_{0}, \mathcal{E}, \mu = \left(\mu_{q}\right)_{q \in \mathcal{Q}}, \theta = \left(\theta_{q}\right)_{q \in \mathcal{Q}}\right)$,
	where $\mathcal{Q}$ denotes a finite set of states, $q_0 \in \mathcal{Q}$ is the start state, $\delta_T: \mathcal{Q} \times \mathcal{Q} \rightarrow \left[0,1\right]$ is a probabilistic transition function, $\mathcal{V}$ is a set of values, $v_0$ is called the start value, $\mathcal{E}$ is a finite set of emissions, $\mu_q$ is a state-specific probability distribution on $\mathcal{E}$, and $\theta_q$ is a state-specific binary operation.
\end{definition}

Similar to the concept of deterministic finite automata, a PAA starts in its start state $q_0$. In contrast to DFAs and DAAs, the transitions are not initiated by input symbols but are purely probabilistic. The transition function $\delta_T$ gives the probability of moving from one state to another. Analogous to DAAs, a PAA performs calculations on $\mathcal{V}$ during a transition. The first value used for such a calculation is the start value $v_0$. The entered state generates an emission according to $\mu_q$ (just as in an HMM) which is used to calculate the resulting value $v_i$ from the former value $v_{i-1}$ according to the binary operation $\theta_q$.

\paragraph*{PAA induced by DAA and Finite-Memory Text Model}

The construction of a PAA from a DAA and a finite-memory text model allows us to compute the probability distribution of the values produced by the DAA on a random text generated by the finite-memory text model.

\begin{definition}[PAA induced by a DAA and a finite-memory text model, \cite{Marschall2012}]\label{def:PAA induced by DAA and text model}
	Let a DAA $\mathcal{D} = \left( \mathcal{Q}, q_{0}, \Sigma, \delta, \mathcal{V}, v_{0}, \mathcal{E}, \left(\eta_{q}\right)_{q \in \mathcal{Q}}, \left(\theta_{q}\right)_{q \in \mathcal{Q}}\right)$ and a finite-memory text model $\mathcal{M} = \left(\mathcal{C},c_0,\Sigma,\varphi\right)$ be given. We define the \emph{PAA induced by $\mathcal{M}$ and $\mathcal{D}$} to have the state space $\mathcal{Q} \times \mathcal{C}$, start state $\left(q_0,c_0\right)$ and a transition function given by $\delta_T\left(\left(q,c\right),\left(q',c'\right)\right) \eqdef \sum_{\sigma \in \Sigma: \delta\left(q,\sigma\right)=q'}{\varphi\left(c,\sigma,c'\right)}$; all other ingredients, i.e.\ value set, start value, emission set and function, as well as operations, are inherited from the DAA.
\end{definition}

In \cite{MarschallR2011} it is shown that a PAA constructed according to Definition \ref{def:PAA induced by DAA and text model} with a DAA $\mathcal{D}$ and a finite-memory text model $\mathcal{M}$ reflects the probabilistic behavior of $\mathcal{D}$ on a random text generated by $\mathcal{M}$. Thus, using the DAA encoding a pattern matching algorithm $\mathcal{A}$ for a pattern $S$, we can compute the character access count distribution of $\mathcal{A}$ on $S$. From this, it is possible to deduce the following space and runtime bounds, which are proven in \cite{MarschallR2011}.

\begin{theorem}[Computing $\mathcal{L}\left(X^{\mathcal{A},S}_n\right)$, \cite{MarschallR2011}]\label{thm:time space construction PAA}
	Let a window based pattern matching algorithm $\mathcal{A}$, a pattern $S$ with $|S| = m$, the functions $g^S_\mathcal{A}$ and $f^S_\mathcal{A}$, a DAA $\mathcal{D}$ with state space $\mathcal{Q}$ encoding $\mathcal{A}$ according to Definition \ref{def:DAA encoding pattern matching algorithm} and a finite-memory text model $\mathcal{M} = \left(\mathcal{C}, c_0, \Sigma, \varphi\right)$ be given. The cost distribution $\mathcal{L}\left(X^{\mathcal{A},S}_n\right)$ of $\mathcal{D}$ for pattern $S$ on a text of length $n$ can be computed using
	$O(n^2 \cdot m \cdot |\mathcal{Q}| \cdot |\mathcal{C}|^2 \cdot |\Sigma|)$ 
	time and 
	$O(|\mathcal{Q}| \cdot |\mathcal{C}| \cdot n \cdot m)$
	space.
\end{theorem}

Since the size of the state space of the DAA calculating the number of character accesses of an algorithm $\mathcal{A}$ on a text $T$ for a pattern $S$ is exponential in the length of $S$, also the space and runtime bounds for computing $\mathcal{L}\left(X^{\mathcal{A},S}_n\right)$ are exponential in the length of $S$.
In the following, we construct a DAA with $O(m^3)$ states.

\section{General construction scheme for DAAs}\label{ch:scheme}
	
In this section, we introduce a general construction scheme for deterministic arithmetic automata that later allows us to construct small DAAs for BM, BMH, B(N)DM, and BOM in a unified manner.
The main idea is to merge states for windows that induce the same shifts and lead to the same costs.
Such sets of windows with the same behavior are represented by only one state.
To this end, we choose a suitable \emph{set of representative strings} $\mathcal{R}$ and define a function $\rep_\mathcal{R}:\Sigma^m\to\mathcal{R}$ that maps each window to its representative.
\begin{definition}[Representative mapping]
For a given set $\mathcal{R}\subset\Sigma^*$ we define $\rep_\mathcal{R}(a)$ to map $a$ to the longest suffix of $a$ that is in $\mathcal{R}$.
\end{definition}
In order for this mapping to be well-defined, we have to ensure that $\mathcal{R}$ contains at least one suffix of every window $a\in\Sigma^m$.
Furthermore, it has to permit a transition function.
Formally, a valid set of window representatives has the following properties:
\begin{definition}[Window representatives]\label{def:window representatives}
A finite set of strings $\mathcal{R}\subset\Sigma^*$ is called \emph{set of window representatives} of length $m$ if
\begin{enumerate}
 \item $\rep_\mathcal{R}$ is well-defined, i.e., every $a\in\Sigma^m$ has a suffix $a_k\dots a_{m-1}$ that is in $\mathcal{R}$ and
 \item there exists a transition function $\delta_\mathcal{R}:\mathcal{R}\times\Sigma\to\mathcal{R}$ such that $\delta_\mathcal{R}(\rep_\mathcal{R}(a),\sigma)=\rep_\mathcal{R}(a_1\dots a_{m-1}\sigma)$ for all $a\in\Sigma^m$ and $\sigma\in\Sigma$.
\end{enumerate}
\end{definition}

Additionally, the set of representatives has to be compatible with shift and cost functions of algorithm $\mathcal{A}$ under study.
\begin{definition}[Compatible window representatives]
Let an algorithm $\mathcal{A}$ and the associated functions $f_\mathcal{A}^S$ and $g_\mathcal{A}^S$ for a pattern $S\in\Sigma^m$ be given.
A set of window representatives $\mathcal{R}$ of length $m$ is called \emph{compatible} with algorithm $\mathcal{A}$ and pattern $S\in\Sigma^m$ if $f_\mathcal{A}^S(a)=f_\mathcal{A}^S(a')$ and $g_\mathcal{A}^S(a)=g_\mathcal{A}^S(a')$ for all $a,a'\in\Sigma^m$ with $\rep_\mathcal{R}(a)=\rep_\mathcal{R}(a')$. For $c \in \mathcal{R}$ define $f_\mathcal{A}^S(c) \eqdef f_\mathcal{A}^S(a)$ and $g_\mathcal{A}^S(c) \eqdef g_\mathcal{A}^S(a)$ for an arbitrary $a\in\Sigma^m$ with $\rep_\mathcal{R}(a)=c$.
\end{definition}

The set of all windows as well as the set of all substrings of the pattern (including the empty string $\varepsilon$) satisfy the properties of a set of window representatives stated in Definition \ref{def:window representatives} as we will show later.

\begin{example}
	Let $\Sigma = \set{a,b}$ and $m=2$ and let the set of all substrings of the pattern $aa$ (including the empty string $\varepsilon$) be the set of window representatives $\mathcal{R}$. Then $\rep_\mathcal{R}(aa)=aa$, $\rep_\mathcal{R}(ba)=a$ and $\rep_\mathcal{R}(bb)=\rep_\mathcal{R}(ab)=\varepsilon$.
\end{example}

Using the notion of window representatives, we define the general construction scheme for deterministic arithmetic automata as

\begin{definition}[Construction scheme for DAAs using window representatives]\label{def:general construction scheme}
	Let $\mathcal{A}$ be an algorithm with associated functions $f_\mathcal{A}^S$ and $g_\mathcal{A}^S$ for a pattern $S \in \Sigma^m$. Let $\mathcal{R}$ be a set of window representatives that is compatible with $\mathcal{A}$ and $S$ with associated transition function $\delta_\mathcal{R}$. We define the general construction scheme for DAAs as
	\begin{itemize}
		\item $\mathcal{Q} \eqdef \mathcal{R} \times \set{0, \dots ,m}$
		\item $q_{0} \eqdef \left(\rep_\mathcal{R}(S),m\right)$
		\item $\delta\left(\left(r,k\right),\sigma\right) \eqdef \begin{cases} \left(\delta_\mathcal{R}(r,\sigma),k-1\right) & \text{if }k>0\\
		\left(\delta_\mathcal{R}(r,\sigma),g^S_{\mathcal{A}}(r)-1\right) & \text{if }k=0\end{cases}$
		\item $\mathcal{V} \eqdef \mathbb{N}$
		\item $v_{0} \eqdef 0$
		\item $\mathcal{E} \eqdef \set{1, \dots ,m}$
		\item $\eta_{\left(r,k\right)} \eqdef \begin{cases}0 & \text{if } k>0\\f^{S}_{\mathcal{A}}(c) & \text{if } k=0\end{cases} ~ ~ ~ ~ ~ \forall ~ (r,k) \in \mathcal{Q}$
		\item $\theta_{\left(r,k\right)}\left(v,e\right) \eqdef v+e ~ ~ ~ ~ ~ \forall ~ (r,k) \in \mathcal{Q}$
	\end{itemize}
\end{definition}
Note that, if we choose all windows to be the set of representatives $\mathcal{R}$ with transition function $\delta_\mathcal{R}(\rep_\mathcal{R}(a),\sigma) = a_1 \dots a_{m-1}\sigma$, the construction scheme yields the same DAA as defined in the direct construction in Definition \ref{def:DAA encoding pattern matching algorithm}.

To prove the correctness of the general construction scheme, we have to show that for an algorithm $\mathcal{A}$ and a pattern $S\in\Sigma^m$ a DAA $\mathcal{D}$ constructed according to it calculates the number of character accesses $\mathcal{A}$ does on a text $T$ for $S$. Since we know from Lemma \ref{lem:DAA calculates number of character comparisons} that a DAA $\mathcal{D'}$ for $\mathcal{A}$ and pattern $S$ constructed according to Definition \ref{def:DAA encoding pattern matching algorithm} calculates the correct result, it suffices to show that $\mathcal{D}$ and $\mathcal{D'}$ calculate the same value. 
%

	Using the previous lemma, we are now able to prove the correctness of the DAA construction scheme in Definition \ref{def:general construction scheme} next.
	
	\begin{theorem}\label{thm:DAA on substrings calculates number of character comparisons}
		Let an algorithm $\mathcal{A}$ and a pattern $S\in\Sigma^m$ be given. Let $\mathcal{D}$ be a DAA constructed according to Definition \ref{def:DAA encoding pattern matching algorithm} for $\mathcal{A}$ and $S$ with state space $\mathcal{Q}$ and transition function $\delta$. Let $\mathcal{D'}$ be constructed according to the general construction scheme from Definition \ref{def:general construction scheme} with the set of window representatives $\mathcal{R}$ for $\mathcal{A}$ and $S$ with transition function $\delta'$. Then $value_\mathcal{D}(T) = value_\mathcal{D'}(T)$ for all $T \in \Sigma^n$.	
	\end{theorem}
	\begin{proof}
 		Let $T \in \Sigma^n$. To prove $value_\mathcal{D}(T) = value_\mathcal{D'}(T)$, we strengthen the statement: Let $\hat{\delta}\left(\left(q_0,v_0\right),T\right) = \left(\left(r,k\right),v\right)$ and $\hat{\delta}'\left(\left(q'_0,v'_0\right),T\right) = \left(\left(w,k'\right),v'\right)$. Then $v=v'$, $r = \rep_\mathcal{R}(w)$ and $k=k'$. This claim can be proven by a straight-forward induction on $n$ using the compatibility of $\mathcal{R}$ with $\mathcal{A}$ and $S$ and the following fact: 
 		For arbitrary states $(w,k)$ and $(r,k)$ from $D$ and $D'$, respectively, with $r=\rep_\mathcal{R}(w)$, $\delta((w,k),\sigma)=(w',k')$ and $\delta'((r,k),\sigma)=(r',k'')$, we have $k'=k''$ by the compatibility of $\mathcal{R}$ with $\mathcal{A}$ and $S$, and \[r' = \delta_\mathcal{R}(r,\sigma) = \delta_\mathcal{R}(\rep_\mathcal{R}(w),\sigma) = \rep_\mathcal{R}(w_1 \dots w_{m-1} \sigma) = \rep_\mathcal{R}(w')\] by the definitions of $\delta$ and $\delta'$ and by the second property of sets of window representatives.
	\end{proof}

\paragraph*{Window representatives for BM, BMH, and B(N)DM}

We claim that the set of all substrings $\subcl{S} = \set{s_i \dots s_j \where 0 \leq i \leq j \leq m-1}$ of a pattern $S\in\Sigma^m$ is a set of window representatives that is compatible with BM, BMH, and B(N)DM. First, we prove that $\subcl{S}$ satisfies the properties of a set of window representatives. Second, we examine the compatibility with BM, BMH, and B(N)DM separately.

\begin{lemma}
Given a pattern $S\in\Sigma^m$, $\subcl{S}$ is a set of window representatives.
\end{lemma}
\begin{proof}
Since the empty string $\varepsilon$ is contained in $\subcl{S}$, $\rep_{\subcl{S}}$ is well-defined. Let $r \in \subcl{S}$ and define $\delta_{\subcl{S}}(r,\sigma)$ to be the longest suffix of $r\sigma$ that is a substring of $S$. With the definition of $\rep_{\subcl{S}}$, $\delta_{\subcl{S}}(r,\sigma) = \rep_{\subcl{S}}(r_1 \dots r_{m-1}\sigma)$ follows directly.
\end{proof}

Next, we consider the compatibility of $\subcl{S}$ with BM, BMH and B(N)DM.

\begin{lemma}\label{lem:compatible with BM}
Let a pattern $S\in\Sigma^m$ and $a, a' \in \Sigma^m$ with $\rep_{\subcl{S}}(a)=\rep_{\subcl{S}}(a')=r$ be given. Then $f^S_{\textnormal{BM}}(a)=f^S_{\textnormal{BM}}(a')$ and $g^S_{\textnormal{BM}}(a)=g^S_{\textnormal{BM}}(a')$.
\end{lemma}
\begin{proof}
	If $a=S$, then also $a'=S$ since $\rep_{\subcl{S}}(a)=\rep_{\subcl{S}}(a')=S$ and therefore $f^S_{\textnormal{BM}}(a)=f^S_{\textnormal{BM}}(a')=m-i+1$. Otherwise, the first mismatch occurs at position $m-i$ with $i \leq |r|$ and hence $a_{m-i+1} \dots a_{m-1} = a'_{m-i+1} \dots a'_{m-1}$. Thus, $f^S_{\textnormal{BM}}(a)=f^S_{\textnormal{BM}}(a')$, $bc^S(a,i)=bc^S(a',i)$ and $gs^S(a,i)=gs^S(a',i)$ since  the functions only depend on the suffixes of length $i-1$ of $a$ and $a'$. Hence, $\max\set{bc^S(a,i),gs^S(w,i)}=\max\set{bc^S(a',i),gs^S(w,i)}$ and therefore $g^S_{\textnormal{BM}}(a)=g^S_{\textnormal{BM}}(a')$.
\end{proof}

\begin{lemma}\label{lem:compatible with BMH}
Let a pattern $S\in\Sigma^m$ and $a, a' \in \Sigma^m$ with $\rep_{\subcl{S}}(a)=\rep_{\subcl{S}}(a')=r$ be given. Then $f^S_{\textnormal{BMH}}(a)=f^S_{\textnormal{BMH}}(a')$ and $g^S_{\textnormal{BMH}}(a)=g^S_{\textnormal{BMH}}(a')$.
\end{lemma}
\begin{proof}
	Since $f^S_{\text{BMH}}(w) = f^S_{\text{BM}}(w)$ for all $w$ and $g^S_{\text{BMH}}$ is a special case of the bad-character rule of the BM algorithm, the claim follows by same reasoning as in Lemma \ref{lem:compatible with BM}.
\end{proof}

\begin{lemma}\label{lem:compatible with B(N)DM}
Let a pattern $S\in\Sigma^m$ and $a, a' \in \Sigma^m$ with $\rep_{\subcl{S}}(a)=\rep_{\subcl{S}}(a')=r$ be given. Then $f^S_{\textnormal{B(N)DM}}(a)=f^S_{\textnormal{B(N)DM}}(a')$ and $g^S_{\textnormal{B(N)DM}}(a)=g^S_{\textnormal{B(N)DM}}(a')$.
\end{lemma}
\begin{proof}
	If $a=S$, then also $a'=S$ since $\rep_{\subcl{S}}(a)=\rep_{\subcl{S}}(a')=S$ and therefore $f^S_{\textnormal{B(N)DM}}(a)=f^S_{\textnormal{B(N)DM}}(a')=m$. Otherwise, $j^S_a$ and $j^S_{a'}$ are the number of transitions the suffix automaton does before entering the FAIL-state (excluding the transition to the FAIL-state) when reading $a$ and $a'$, respectively. Thus, $a_{j^S_a} \dots a_{m-1}=a'_{j^S_{a'}} \dots a'_{m-1}=r$ and hence $j^S_a=j^S_{a'}$, i.e. $f^S_{\textnormal{B(N)DM}}(a)=f^S_{\textnormal{B(N)DM}}(a')$. Furthermore, $\mathcal{I}^S(a)$ contains all positions $i$ such that the suffix automaton is in an accepting state after reading $i$ characters of $a$. Since the FAIL-state is a non-accepting sink, $\max\set{i \where i \in \mathcal{I}^S(a)} = \max\set{i \where i \in \mathcal{I}^S(a')} < j^S_a = j^S_{a'}$ and thus $g^S_{\textnormal{B(N)DM}}(a)=g^S_{\textnormal{B(N)DM}}(a')$.
\end{proof}

Since we showed that the set of all substrings of the pattern $S$ is a set of window representatives that is compatible with BM, BMH, and B(N)DM, we can use the general construction scheme from Definition \ref{def:general construction scheme} with $\subcl{S}$ to construct DAAs. Figure \ref{fig:DAA reduced} shows the states and the transition function of a DAA for BMH constructed according to the scheme. Since $|\subcl{S}| \leq m^2$ for every pattern of length $m$, the size of the state space of the DAAs for BM, BMH, and B(N)DM is bounded by $O(m^3)$.
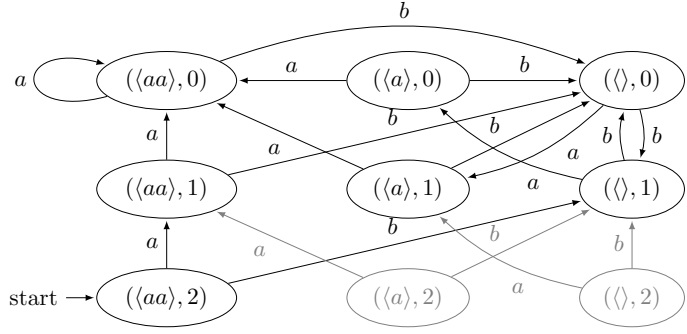
\begin{figure}
	\centering
	\scalebox{0.8}{
	\begin{tikzpicture}[->,>=latex,shorten >=1pt,auto,node distance=1.8cm and 1.8cm,thin]
		  \tikzstyle{every state}=[fill=none,draw=black,text=black]
		  \node[state,ellipse] (AA0)                    		{$(\class{aa},0)$};
		  \node[state,ellipse] (AA1)			[below of=AA0]	{$(\class{aa},1)$};
		  \node[state,ellipse,initial] (AA2)			[below of=AA1]	{$(\class{aa},2)$};
		  \node[state,ellipse] (A0)            [right=of AA0]	{$(\class{a},0)$};
		  \node[state,ellipse] (A1)				[below of=A0]	{$(\class{a},1)$};
		  \node[state,ellipse,text=gray,draw=gray] (A2)				[below of=A1]	{$(\class{a},2)$};
		  \node[state,ellipse] (E0)            [right=of A0]	{$(\class{ },0)$};
		  \node[state,ellipse] (E1)				[below of=E0]	{$(\class{ },1)$};
		  \node[state,ellipse,text=gray,draw=gray] (E2)				[below of=E1]	{$(\class{ },2)$};

		   \path 	(AA0)	edge	[loop left]		node 			{$a$}	(AA0)
		   					edge	[bend left=20]	node			{$b$}	(E0)
		   			(AA1)	edge					node			{$a$}	(AA0)
		   					edge					node			{$b$}	(E0)
		   			(AA2)	edge					node			{$a$}	(AA1)
		   					edge					node			{$b$}	(E1)
		   			(A0)	edge	[above]			node			{$a$}	(AA0)
		   					edge					node			{$b$}	(E0)
		   			(A1)	edge					node			{$a$}	(AA0)
		   					edge					node	[pos=0.4] 		{$b$}	(E0)
		   			(A2)	edge	[gray]			node	[pos=0.6]		{$a$}	(AA1)
		   					edge	[gray]			node	[pos=0.4]		{$b$}	(E1)
		   			(E0)	edge	[bend left=15]	node	[pos=0.35]		{$a$}	(A1)
		   					edge	[bend left=15]	node			{$b$}	(E1)
		   			(E1)	edge	[bend left=15]	node	[pos=0.2]		{$a$}	(A0)
		   					edge	[bend left=15]	node			{$b$}	(E0)
		   			(E2)	edge	[bend left=15,gray]	node	[pos=0.3]		{$a$}	(A1)
		   					edge	[gray]			node			{$b$}	(E1);
	\end{tikzpicture}}
	\caption{States and transition function for the DAA with reduced state space modelling the BMH algorithm for $\Sigma=\set{a,b}$ and $S=aa$. Unreachable states and their outgoing edges are marked in gray.}\label{fig:DAA reduced}
\end{figure}

\paragraph*{Window representatives for BOM}

Since the factor-oracle accepts strings that are no substrings of the pattern, $\subcl{S}$ is not compatible with BOM. We show that we can use the set of the reversion of all strings that are accepted by the factor oracle, $\focl{S}$, instead:
\[\focl{S} \eqdef \set{x \where  \text{the factor oracle of } S^\textnormal{rev} \text{ is not in the FAIL-State after reading } x^\text{rev}}\]

\begin{lemma}
Given a pattern $S\in\Sigma^m$, $\focl{S}$ is a set of window representatives.
\end{lemma}
\begin{proof}
Since the empty string $\varepsilon$ is accepted by the factor oracle of $S^\textnormal{rev}$, $\rep_{\focl{S}}$ is well-defined. Let $r \in \focl{S}$ and define $\delta_{\focl{S}}(r,\sigma)$ to be the longest suffix of $r\sigma$ that is accepted by the factor oracle of $S^\text{rev}$. Then $\delta_{\focl{S}}(r,\sigma) = \rep_{\focl{S}}(r_1 \dots r_{m-1}\sigma)$ follows directly with the definition of $\rep_{\focl{S}}$.
\end{proof}

\begin{lemma}\label{lem:compatible with BOM}
Let a pattern $S\in\Sigma^m$ and $a, a' \in \Sigma^m$ with $\rep_{\focl{S}}(a)=\rep_{\focl{S}}(a')=r$ be given. Then $f^S_{\textnormal{BOM}}(a)=f^S_{\textnormal{BOM}}(a')$ and $g^S_{\textnormal{BOM}}(a)=g^S_{\textnormal{BOM}}(a')$.
\end{lemma}
\begin{proof}
If $a=S$, then also $a'=S$ since $S \in \focl{S}$ and $\rep_{\focl{S}}(a)=\rep_{\focl{S}}(a')=S$. Hence, $f^S_\text{BOM}(a)=f^S_\text{BOM}(a')=m$ and $g^S_\text{BOM}(a)=g^S_\text{BOM}(a')=1$. Otherwise, $k^S_a$ and $k^S_{a'}$ are the number of transitions the factor oracle does before entering the FAIL-state when reading $a$ and $a'$, respectively. Thus, $a_{k^S_a} \dots a_{m-1}=a'_{k^S_{a'}} \dots a'_{m-1}=r$ and hence $k^S_a=k^S_{a'}$. Therefore $f^S_{\textnormal{BOM}}(a)=f^S_{\textnormal{BOM}}(a')$ and $g^S_{\textnormal{BOM}}(a)=g^S_{\textnormal{BOM}}(a')$.
\end{proof}

The size of the state space of a DAA constructed according to the general scheme for BOM depends on the factor oracle of $S^\text{rev}$. It accepts more strings than all substrings of $S^\text{rev}$ but in practice not many more \cite{NavarroR2002}.
			
\section{Construction of a PAA from a DAA with reduced state space}\label{ch:PAA}
Since we showed that the DAAs constructed in Section~\ref{ch:scheme} calculate the number of character accesses of BM, BMH, B(N)DM, and BOM, respectively, we can use the construction given in Definition \ref{def:PAA induced by DAA and text model} to obtain a PAA for the character access count distribution.
Figure~\ref{fig:PAA} shows an example of a resulting PAA for the BMH algorithm.
\begin{figure}
	\centering
	\begin{minipage}[t]{0.3\linewidth}
		\scalebox{0.7}{
		\begin{tikzpicture}[->,>=latex,shorten >=1pt,auto,node distance=1.8cm and 1.8cm,thin]
			  \tikzstyle{every state}=[fill=none,draw=black,text=black]
			  
			  \node[state,ellipse,initial] (C0)          		{$c_0$};
			  \node[state,ellipse] (C1)			[right=of C0]	{$c_1$};
			  
			 \path 	(C0)	edge	[loop above]	node 			{$0.4: a$}	(C0)
		   					edge	[bend left=20]	node			{$0.6: b$}	(C1)
		   			(C1)	edge	[bend left=20]	node			{$0.8: a$}	(C0)
		   					edge	[loop above]	node			{$0.2: b$}	(C1);
			  
		\end{tikzpicture}}
		\caption{A finite-memory text model over the alphabet $\Sigma=\set{a,b}$ with two states.}\label{fig:text model}
	\end{minipage}
	\hfill
	\begin{minipage}[t]{0.66\linewidth}
		\scalebox{0.69}{
		\begin{tikzpicture}[->,>=latex,shorten >=1pt,auto,node distance=1.8cm and 1.5cm,thin]
			  \tikzstyle{every state}=[fill=none,draw=black,text=black]
			  
			  \node[state,ellipse] (AA0)                    		{$((\class{aa},0),c_0)$};
			  \node[state,ellipse] (AA1)			[below of=AA0]	{$((\class{aa},1),c_0)$};
			  \node[state,ellipse,initial] (AA2)	[below of=AA1]	{$((\class{aa},2),c_0)$};
			  \node[state,ellipse] (A0)            [right=of AA0]	{$((\class{a},0),c_0)$};
			  \node[state,ellipse] (A1)				[below of=A0]	{$((\class{a},1),c_0)$};
			  \node[state,ellipse] (E0)            [right=of A0]	{$((\class{ },0),c_1)$};
			  \node[state,ellipse] (E1)				[below of=E0]	{$((\class{ },1),c_1)$};
			  
			 \path 	(AA0)	edge	[loop left]		node 			{$0.4$}	(AA0)
		   					edge	[bend left=20]	node			{$0.6$}	(E0)
		   			(AA1)	edge					node			{$0.4$}	(AA0)
		   					edge					node			{$0.6$}	(E0)
		   			(AA2)	edge					node			{$0.4$}	(AA1)
		   					edge	[below]			node			{$0.6$}	(E1)
		   			(A0)	edge	[above]			node			{$0.4$}	(AA0)
		   					edge					node			{$0.6$}	(E0)
		   			(A1)	edge					node			{$0.4$}	(AA0)
		   					edge					node	[pos=0.4] 		{$0.6$}	(E0)
		   			(E0)	edge	[bend left=15]	node	[pos=0.35]		{$0.8$}	(A1)
		   					edge	[bend left=15]	node			{$0.2$}	(E1)
		   			(E1)	edge	[bend left=15]	node	[pos=0.2]		{$0.8$}	(A0)
		   					edge	[bend left=15]	node			{$0.2$}	(E0);
			  
		\end{tikzpicture}}
		\caption{States and transition function for the PAA with reduced state space modelling the BMH algorithm for $\Sigma=\set{a,b}$ and $S=aa$ and the text model shown in Figure \ref{fig:text model}. Unreachable states are omitted.}\label{fig:PAA}
	\end{minipage}
\end{figure}

For the DAAs modeling the BM, BMH and B(N)DM algorithms, the state space of the resulting PAA is bounded by $O(m^3 \cdot |\mathcal{C}|)$, where $\mathcal{C}$ denotes the state space of the finite-memory text model used. For the DAA for the BOM algorithm, the state space of the resulting PAA is bounded by $O(m \cdot |\focl{S}| \cdot |\mathcal{C}|)$.
As a direct consequence of Theorem~\ref{thm:time space construction PAA}, we obtain the following runtimes for computing the respective character access count distributions.

\begin{theorem}\label{thm:final_runtimes}
Let a finite-memory text model $\mathcal{M} = \left(\mathcal{C}, c_0, \Sigma, \varphi\right)$ and a pattern $S$ with $|S|=m$ be given. 
For $\mathcal{A}\in\{\text{BM}, \text{BMH}, \text{B(N)DM}\}$, the cost distributions $\mathcal{L}(X^{\mathcal{A},S}_n)$ can be computed using
$O(n^2 \cdot m^4 \cdot |\mathcal{C}|^2 \cdot |\Sigma|)$  time and  $O(m^4 \cdot |\mathcal{C}| \cdot n)$ space.
For the BOM algorithm, $\mathcal{L}(X^{\textnormal{BOM},S}_n)$ can be computed using 
$O(n^2 \cdot m^2 \cdot |\focl{S}| \cdot |\mathcal{C}|^2 \cdot |\Sigma|)$ time and $O(m^2 \cdot |\focl{S}| \cdot |\mathcal{C}| \cdot n)$ space.
\end{theorem}

\section{Conclusions}\label{sec:conclusion}
We introduced a construction mechanism for state spaces to analyze pattern matching algorithms.
It allows to construct DAAs of size $O(m^3)$ for the Boyer-Moore, Boyer-Moore-Horspool, and Backward (Non-Deterministic) DAWG Matching algorithms as well as DAAs of size $O(m \cdot |\focl{S}|)$ for the Backward Oracle Matching  algorithm.
Building on the framework of deterministic/probabilistic arithmetic automata~\cite{Marschall2012,MarschallR2011}, we immediately obtain an algorithm to compute the distribution of character accesses for general finite-memory text models, including i.i.d.\ models, Markov models of arbitrary order, and character-emitting hidden Markov models.
For BM, BMH, and B(N)DM, this algorithm runs in polynomial time (Theorem~\ref{thm:final_runtimes}) and we are not aware of prior algorithms to achieve this  (although for BMH a polynomial-time algorithm for the case of first-order Markovian models was known, \cite{Tsai2006}).

The PAAs we construct could also be leveraged to obtain asymptotic results, a direction we have not explored here.
Under mild additional assumptions, the Markov chain defined by state space and transition function of the PAA converge to a steady state, and the equilibrium state distribution can straightforwardly be used to compute $\lim_{n\to\infty}\frac{E[X^{\mathcal{A},S}_n]}{n}$.
Similar to \cite{Tsai2006}, we might also show that (appropriately centered and normalized) $X^{\mathcal{A},S}_n$ follows a normal distribution asymptotically and compute mean and variance.
We plan to work out this connection formally in the future.

The question of lower bounds for the DAA sizes is also open.
We consider investigating the sizes of minimal DAAs an interesting future endeavor.

\bibliographystyle{plain}
\bibliography{bibliography}

\end{document}